\documentclass[conference]{IEEEtran}
\IEEEoverridecommandlockouts

\usepackage{amsmath,amssymb}
\usepackage{graphicx}
\usepackage{dsfont}
\usepackage{framed}
\usepackage{amsthm}

\newcommand{\F}{\mathds{F}}
\newcommand{\supp}{\mathsf{supp}}
\newtheorem{thm}{Theorem}
\newtheorem{lem}[thm]{Lemma}

\newtheorem{prop}[thm]{Proposition}
\newtheorem{coro}[thm]{Corollary}

\newtheorem{defn}[thm]{Definition}

\newtheorem{prob}[thm]{Problem}
\newtheorem{ques}[thm]{Question}

\newcommand{\zo}{\{0,1\}}
\newcommand{\eps}{\epsilon}
\newcommand{\poly}{\mathsf{poly}}

\newcommand{\CC}{\mathds{C}}
\newcommand{\Rip}{\textrm{RIP}}
\newcommand{\RIP}[1]{\textrm{RIP}{\text{-}}{#1}}
\newcommand{\cL}{\mathcal{L}}
\newcommand{\cD}{\mathcal{D}}
\newcommand{\C}{\mathcal{C}}
\newcommand{\cB}{\mathcal{B}}

\newcommand{\Z}{\mathds{Z}}
\newcommand{\dist}{\mathsf{dist}}

\newcommand{\Iprod}[2]{\left\langle #1, #2 \right\rangle}
\newcommand{\innr}[1]{\langle #1 \rangle}
\newcommand{\Innr}[1]{\left\langle #1 \right\rangle}

\newcommand{\Bool}{\mathsf{Bool}}
\newcommand{\Sph}{\mathsf{Sph}}
\newcommand{\cOne}{\C/\boldsymbol{1}}
\newcommand{\cX}{\mathcal{X}}
\newcommand{\cY}{\mathcal{Y}}
\newcommand{\E}{\mathbb{E}}

\begin{document}
\setlength{\pdfpageheight}{\paperheight}
\setlength{\pdfpagewidth}{\paperwidth}

\title{Coding-Theoretic Methods for Sparse Recovery} \author{
  \authorblockN{Mahdi Cheraghchi\thanks{} \authorblockA{Computer
      Science Department \\ Carnegie Mellon University \\ Pittsburgh,
      PA, USA \thanks{Email:
        $\langle$cheraghchi@cmu.edu$\rangle$. Research supported by
        the Swiss National Science Foundation research grant
        PBELP2-133367.} }}}

\maketitle

\begin{abstract}
  We review connections between coding-theoretic objects and sparse
  learning problems. In particular, we show how seemingly different
  combinatorial objects such as error-correcting codes, combinatorial
  designs, spherical codes, compressed sensing matrices and group
  testing designs can be obtained from one another. The reductions
  enable one to translate upper and lower bounds on the parameters
  attainable by one object to another.  We survey some of the
  well-known reductions in a unified presentation, and bring some
  existing gaps to attention. New reductions are also introduced; in
  particular, we bring up the notion of minimum \emph{$L$-wise
    distance} of codes and show that this notion closely captures the
  combinatorial structure of RIP-2 matrices. Moreover, we show how
  this weaker variation of the minimum distance is related to
  combinatorial list-decoding properties of codes.
\end{abstract}

\section{Introduction}

Consider an $n$-dimensional vector $x \in \CC^N$ that is $L$-sparse,
i.e., has $L$ or less non-zero entries. The basic goal in compressed
sensing is to design a \emph{measurement matrix} $M \in \CC^{n \times
  N}$ such that from the measurement outcome
\[
y := M \cdot x \in \CC^n
\]
it is information-theoretically possible to uniquely reconstruct $x$.
Since $x$ can be described by up to $L$ complex numbers plus $L$
integers in $[N]:=\{1,\ldots,N\}$ (that describe the support of the
vector), it is natural to expect that the amount of measurements $n$
can be made substantially less than the dimension $N$ of the vector,
even if one uses a set of linear forms as above to encode $x$.  It
turns out that the above intuition can be formalized and indeed there
are measurement matrices with significantly smaller number of rows
than columns \cite{Candes1,Candes,Candes2,Donoho,Candes-Stable}. In
fact one can even obtain $n=2L$ by taking $M$ to be a
\emph{Vandermonde matrix} \cite{Tarokh}.

Similar to compressed sensing, one can think of different \emph{sparse
  recovery} problems with the goal of identifying objects that are
known to have sparse representations. For example, compressed sensing
can be extended to vectors over finite fields, which makes it
essentially equivalent to the well-studied \emph{syndrome decoding}
problem of error-correcting codes, or to non-linear measurements. A
particularly interesting class of non-linear measurements is
characterized by \emph{disjunctions}, which gives rise to a class of
sparse recovery problems known as (non-adaptive) \emph{combinatorial
  group testing} (cf.\ \cite{ref:groupTesting,ref:DH06}).  In group
testing, the measurement matrix and the sparse vector $x$ both lie in
the Boolean domain $\zo$. Then, the $i$th entry of the measurement $y$
is defined as the logical expression
\[
y(i) := (M_{i,1} \land x_1) \lor (M_{i,2} \land x_2) \lor \cdots \lor
(M_{i,n} \land x_n),
\]
where $M_{i,j}$ denotes the $j$th entry of the $i$th row of $M$.  Same
as compressed sensing, group testing measurement matrices are known
for $n \ll N$.

Even though we have defined the sparse recovery problems above in the
most basic combinatorial form, in practice it is desirable to have
measurement matrices with further qualities. For example, it is
desirable to have an explicit construction of the measurement matrix;
e.g., a polynomial-time algorithm for computing the entries of the
matrix. Moreover, the \emph{decoding algorithm} to infer the sparse
vector from the measurement outcomes is of crucial importance and it
is desirable to have as efficient a decoder as possible. Third,
imprecisions are inevitable in practice and the design should be
robust in presence of errors.

Going through the vast amount of literature in sparse recovery makes
it evident that the theory of error-correcting codes proves to be of
central importance in addressing the three basic requirements
above. In this work, we revisit and highlight some of the known
connections between coding theory and sparse recovery in a unified
exposition, and moreover we introduce new connections.  In particular,
we study connections between coding-theoretic objects such as codes
with large distance, list-decodable codes, combinatorial designs, and
spherical codes to sparse recovery problems.

Coding theoretic methods have also been successfully applied to other
sparse recovery problems, such as extensions of group testing to the
threshold model and learning sparse hypergraphs, as well as low-rank
matrix completion problems.  However, due to the space limit, in this
presentation we will only focus on the basic problems of compressed
sensing and (noiseless) group testing.  Moreover, we will only be able
to emphasize on a few of the most basic reductions from
coding-theoretic objects to measurement designs, and vice versa.

The rest of the paper is organized as follows. In
Section~\ref{sec:notation} we review the notation that we use
throughout the paper. Then, in Section~\ref{sec:combinatorics} we
introduce the notions of \emph{Restricted Isometry Property} ($\Rip$)
and disjunct matrices that are central to compressed sensing and group
testing, respectively. Section~\mbox{\ref{sec:mindist}} shows how the
minimum distance of error-correcting codes relate to the $\Rip$.
Section~\ref{sec:avgToRIP} introduces the new idea of extending the
notion of the minimum distance of codes to tuples of codewords, as
opposed to pairs. Then, we show a new result that this notion is more
closely related to the $\Rip$ than the minimum
distance. Section~\ref{sec:design} shows the relationship between
codes, combinatorial designs, and group testing
schemes. Section~\ref{sec:listdecoding} touches upon some new
connections between $\Rip$ matrices and list-decodable codes. Finally,
Section~\ref{sec:conclusion} concludes the work with possible future
directions.

\subsection{Notation} \label{sec:notation}

For a vector $v = (v_1, \ldots, v_n)$, we use the convention $v(i) :=
v_i$ for the $i$th entry of $v$ and define $\supp(v) \subseteq [n]$ to
denote the \emph{support} of $v$.  For an $n \times N$ matrix $M$, and
a subset of column indices $\cL \subseteq [N]$, the submatrix of $M$
obtained by removing all columns of $M$ outside $\cL$ is denoted by
$M|_\cL$. For a complex vector $v$, the $\ell_p$ norm of $v$ is
denoted by $\| v \|_p$. When $p=2$, we may omit the subscript and
simply write $\| v \|$. For a complex number $a \in \C$, the conjugate
of $a$ is denoted by $a^\ast$. For the most part in this write-up, we
assume without loss of generality that $q$-ary codes are defined over
the alphabet $\Z_q$ even if we do not use the ring structure of
$\Z_q$.  For Boolean vectors $x$ and $y$, we use $\Delta(x,y)$ to
denote the Hamming distance between $x$ and $y$.

The \emph{statistical distance} between two distributions $\cX$ and
$\cY$ with probability measures $\Pr_\cX(\cdot)$ and $\Pr_\cY(\cdot)$
defined on the same finite space $\Sigma$ is given by $ \frac{1}{2}
\sum_{s \in \Sigma} |\Pr_\cX(s) - \Pr_\cY(s)|, $ which is half the
$\ell_1$ distance of the two distributions when regarded as vectors of
probabilities over $\Sigma$. Two distributions $\cX$ and $\cY$ are
said to be $\eps$-close if their statistical distance is at most
$\eps$.

\section{Combinatorics of Sparse Recovery} \label{sec:combinatorics}

It is easy to see that for the purpose of compressed sensing, a
measurement matrix $M$ can distinguish between all $L$-sparse vectors
iff for every subset $\cL$ of up to $2L$ columns, the right kernel of
the sub-matrix $M|_\cL$ is zero.  This condition is in particular
achieved by Vandermonde matrices \cite{Tarokh}.  However, in general
such matrices need not be well-conditioned in the sense that the
action of the matrix on sparse vectors may greatly affect their norm,
which is not desirable in presence of imprecisions and/or noise in the
measurements.  A stronger condition would be to require each
sub-matrix $M|_\cL$ to be nearly orthogonal.  This gives rise to the
notion of \emph{Restricted Isometry Property} ($\Rip$) as defined
below.

\begin{defn} \label{def:RIP} Let $p, \alpha > 0$ be real parameters.
  An $n \times N$ matrix $M \in \CC^{n \times N}$ is said to satisfy
  $\RIP{p}$ of order $L$ with constant $\alpha$ (or said to have
  $L$-$\RIP{p}$, in short) if for every $\cL \subseteq [N]$ with
  $|\cL| \leq L$ and every column vector $x \in \CC^{|\cL|}$, we have
  $ (1-\alpha) \| x \|_p \leq \| M|_\cL \cdot x \|_p \leq (1+\alpha)
  \| x \|_p.  $ The constant $\alpha$ is sometimes omitted, in which
  case it is implicitly assumed to be an absolute constant in $(0,1)$.
\end{defn}

In this work, we will focus on the special case $p=2$. In this case,
it is known that an $\Rip$ matrix is sufficient for distinguishing
between sparse vectors even in presence of noise and when the vector
being measured is \emph{approximately sparse}
(cf.~\cite{CandesRIP,CandesStable,Candes1}).  Moreover, a linear
program can be used to reconstruct the sparse vector.  Similar (but
weaker) results are known about the $\RIP{1}$ (cf.~\cite{ref:Sina}).

For group testing the following basic notion turns out to exactly
capture the combinatorial structure needed for distinguishing between
$L$-sparse vectors (cf.~\cite{ref:groupTesting}):

\begin{defn} \label{def:disjunct} An $n \times N$ binary matrix is
  called \emph{$L$-disjunct} if for any choice of $L+1$ columns $M_0,
  \ldots, M_L$ of the matrix, we have $ \bigcup_{i \in [L]} \supp(M_i)
  \nsubseteq \supp(M_0).  $
\end{defn}

\section{From Minimum Distance to RIP} \label{sec:mindist}

In this section we describe a few well known results about
construction of RIP matrices from codes with good minimum distance
properties. These techniques are used, for example, in
\cite{ref:DV07,ref:HCS08,ref:AHSC09} for deterministic construction of
$\Rip$ matrices from specific families of codes.  The reductions are
based on the following simple embeddings of finite-domain vectors into
the complex domain:

\begin{defn} \label{def:Sph} Let $c \in \Z_q^n$ be a $q$-ary vector.

  \begin{enumerate}
  \item Let $\zeta \in \CC$ be a primitive $q$th root of unity.  The
    \emph{spherical embedding} of $c$, denoted by $\Sph(c)$, is a
    vector $c' \in \CC^n$ where for each $i \in [n]$, we define $c'(i)
    := \zeta^{c(i)}/\sqrt{n}$.

  \item For any $i \in \Z_q$, denote by $e_i$ the $i$th standard basis
    vector in $\zo^q$. That is, $e_i(j)=1$ if $j=i+1$ and $e_i(j)=0$
    if $j \neq i+1$.  The \emph{Boolean embedding} of $c$, denoted by
    $\Bool(c)$, is a vector $c'' \in \zo^{qn}$ obtained from $c$ by
    replacing each element $c(i)$ of $c$ with the $q$-dimensional
    vector $e_{c(i)}$.
  \end{enumerate}
\end{defn}

For example, consider the $4$-dimensional binary vector $c := (0, 1,
1, 0) \in \F_2^4$. Then, we have $\zeta=-1$ and
\[
\begin{array}{c}
  \Sph(c) = (1, -1, -1, 1) \\
  \Bool(c) = (1,0, 0,1, 0,1, 1,0).
\end{array}
\]

The property that is later needed for the $\Rip$ constructions is the
\emph{bias} of the code, defined below.

\begin{defn}
  A vector $c \in \Z_q^n$ naturally induces a probability measure
  $\mu_c$ on the alphabet $\Z_q$, where for each $i\in \Z_q$,
  $\mu_c(i)$ is the fraction of coordinate positions at which $c$ is
  equal to $i$. The vector $c$ is said to be $\eps$-biased if $\mu_c$
  is $\eps$-close to uniform.
\end{defn}

\begin{defn}
  A code $\C \subseteq \Z_q^n$ is said to be $\eps$-biased if, for
  every pair of distinct codewords $c, c' \in \C$, the difference
  vector $c - c'$ is $\eps$-biased.
\end{defn}

Even though small bias is in general stronger than large minimum
distance, for \emph{balanced codes} as defined below the two notions
are essentially equivalent, up to simple manipulations of the code.

\begin{defn}
  A (possibly non-linear) code $\C \subseteq \F_q^n$ is called
  \emph{balanced} if, for every $c \in \C$, and every $\alpha \in
  \F_q$, $c + \alpha \boldsymbol{1} \in \C$, where $\boldsymbol{1}$
  denotes the all-ones vector.
\end{defn}

\begin{defn}
  Let $\C \subseteq \Z_q^n$ be a balanced code.  Consider the
  equivalence relation between codewords that differ by a multiple of
  $\boldsymbol{1}:=(1, \ldots, 1)$.  This partitions the codewords of
  $\C$ into equivalence classes.  Define $\cOne$ to be any sub-code of
  $\C$ that picks exactly one codeword from each equivalence class.
\end{defn}

\begin{prop} \label{prop:epsbias} Let $\C \subseteq \Z_q^n$ be a
  balanced code with relative minimum distance at least
  $1-(1+\eps)/q$. Then the sub-code $\cOne$ is $\eps$-biased.
\end{prop}

\begin{proof}
  Consider any pair of distinct codewords $c, c' \in \cOne$ and define
  $C' := \{ c' + \alpha \boldsymbol{1}\colon \alpha \in \Z_q\}$.
  Since $\C$ is balanced, $C' \subseteq \C$. Moreover, $c \notin C'$,
  and therefore, the relative Hamming distance between $c$ and any
  codeword in $C'$ is at least $1-(1+\eps)/q$. In particular, the
  fraction of position at which $c-c'$ is equal to any value $\alpha
  \in \Z_q$ is at most $(1+\eps)/q$ (since otherwise, the distinct
  vectors $c$ and $c' - \alpha \boldsymbol{1}$ would agree at more
  than $(1+\eps)/q$ fraction of the positions, violating the minimum
  distance property). From the definition of statistical distance, we
  conclude that $c-c'$ is $\eps$-biased.
\end{proof}

Now we are ready to describe how small bias is related to geometric
properties of the complex embeddings in Definition~\ref{def:Sph}.

\begin{prop} \label{prop:diffBias} Suppose $c, c' \in \Z_q^n$ are so
  that $c-c'$ is $\eps$-biased. Then $|\innr{\Sph(c),\Sph(c')}| \leq
  2\eps$.
\end{prop}

\begin{proof}
  For $i \in \Z_q$, let $p_i := |\{ j\colon c(j)-c'(j)=i \}|/n$. We
  know that the values $p_i$ induce a probability distribution on
  $\Z_q$ that is $\eps$-close to uniform. Define $s_1 := \Sph(c)$ and
  $s_2 := \Sph(c')$. We have that
  \begin{eqnarray}
    |\innr{s,s'}| &=& \big| \sum_{i \in [n]} s_1(i) s_2^\ast(i)\big| \nonumber \\
    &=& \big| \sum_{i \in [n]} \zeta^{c(i)-c'(i)}/n \big| 
    = n \big| \sum_{i \in \Z_q} p_i \zeta^{i}/n \big| \nonumber \\
    &=& \big| \sum_{i \in \Z_q} (p_i- 1/q) \zeta^{i} \big| \label{eqn:addzero} \\
    &\leq& \sum_{i \in \Z_q} |p_i- 1/q| \leq 2 \eps, \nonumber  
  \end{eqnarray}
  where \eqref{eqn:addzero} is due to the fact that $\sum_{i \in \Z_q}
  \zeta^i = 0$.
\end{proof}

\begin{defn}
  Let $\C \subseteq \Z_q^n$ be a code.
  \begin{enumerate}
  \item The \emph{spherical embedding} of $\C$ is a complex $n \times
    |\C|$ matrix with columns indexed by the elements of $\C$.  The
    column corresponding to a codeword $c \in \C$ is $\Sph(c)$.

  \item The \emph{Boolean embedding} of $\C$ is a real $n \times |\C|$
    matrix with $0/1$ entries and columns indexed by the elements of
    $\C$.  The column corresponding to a codeword $c \in \C$ is
    $\Bool(c)$.
  \end{enumerate}
\end{defn}

\begin{defn}
  A set $\C \in \CC^n$ is a \emph{spherical
    code}\footnote{Traditionally spherical codes are defined under the
    constraint of having upper bounded (but possibly negative) mutual
    inner products.  In this work we will require them to have low
    \emph{coherence}, which is a stronger property.} if each $c \in
  \C$ satisfies $\| c \| = 1$. Moreover, $\C$ is said to be
  \emph{$\eps$-coherent} if, for any distinct $c, c' \in \C$, we have
  $|\innr{c,c'}| \leq \eps$.
\end{defn}

Using the above definition, Proposition~\ref{prop:diffBias}
immediately implies the following.

\begin{coro} \label{coro:biasToCoherence} Let $\C \subseteq \Z_q^n$ be
  an $\eps$-biased code. Then the column set of $\Sph(\C)$ forms a
  spherical code with coherence at most $2\eps$.
\end{coro}

We are now ready to state how low-coherent spherical codes are related
to $\Rip$ matrices.  This is shown in the following well-known
proposition:

\begin{prop} \label{prop:incoherent} Suppose that the column set of an
  $n \times N$ complex matrix $M$ form an $\eps$-coherent spherical
  code. Then, $M$ satisfies $\RIP{2}$ of order $L$ with constant $L
  \eps$.
\end{prop}

\begin{proof}
  Consider an $n \times L$ sub-matrix $M'$ of $M$ where $M' := (M'_1
  \mid \cdots \mid M'_L)$ and the $M'_i$ are unit vectors in $\CC^n$,
  and let $x=(x_1, \ldots, x_L) \in \CC^L$. We can write
  \begin{gather*}
    \| M' x \|^2 - \| x \|^2 = \innr{M'x, M'x} - \| x \|^2 \\
    = \innr{\sum_{i \in [L]} x_i M'_i, \sum_{i \in [L]} x_i M'_i} - \| x \|^2 \\
    = \sum_{i \in [L]} x_i^2 \| M_i \|^2 + \sum_{\substack{i, j \in
        [L] \\ i \neq j}}x_i x_j \innr{M'_i, M'_j}- \| x \|^2
    \\
    = \sum_{\substack{i, j \in [L] \\ i \neq j}}x_i x_j \innr{M'_i,
      M'_j} =: \eta.
  \end{gather*}
  And now we have
  \[
  |\eta| \leq \eps |\sum_{i,j} x_i x_j| \leq \eps (\sum_{i\in [L]}
  x_i)^2 \leq \eps \| x \|_1^2 \leq L\eps \|x\|_2^2,
  \]
  where the last inequality is by Cauchy-Schwarz.
\end{proof}

The above proposition can be combined with
Proposition~\ref{prop:epsbias} and
Corollary~\ref{coro:biasToCoherence} to show the following result.

\begin{coro} \label{coro:sphRIP} Let $\C \subseteq \Z_q^n$ be a
  balanced code with relative minimum distance at least
  $1-(1+\eps)/q$. Then, $\Sph(\cOne)$ satisfies $\RIP{2}$ of order $L$
  with constant $2L\eps$.
\end{coro}

As for Boolean embedding $\Bool(\cdot)$, the following observation is
easy to verify:

\begin{prop}
  Let $\C \subseteq \Z_q^n$ be a code with relative minimum distance
  at least $\delta$.  Then, columns of $\Bool(\C)/\sqrt{n}$ form a
  $(1-\delta)$-coherent spherical code.
\end{prop}

Combined with Proposition~\ref{prop:incoherent}, we see that Boolean
embedding can also result in $\Rip$ matrices.

\begin{coro} \label{coro:boolRIP} Let $\C \subseteq \Z_q^n$ be a code
  with relative minimum distance at least $1-(1+\eps)/q$.  Then,
  $\Bool(\C)/\sqrt{n}$ satisfies $\RIP{2}$ of order $L$ with constant
  $(1+\eps)L/q$.
\end{coro}

Now we consider instantiations of the above result with asymptotically
good families of codes.  Various positive and negative bounds are
known for rate-distance trade-offs achievable by error-correcting
codes. On the positive side, the \emph{Gilbert-Varshamov bound} on
codes \cite{gilbert,varshamov} states that for every alphabet size $q
> 1$ and constant $\delta \in [0,1-1/q)$, there are $q$-qry codes with
rate
\begin{equation} \label{eqn:GV} R \geq
  1-h_q(\delta)-o(1), \end{equation} where $h_q(\cdot)$ is the $q$-ary
entropy function defined as
\begin{equation*} \label{eqn:HQ} h_q(\delta) := \delta \log_q(q-1) -
  \delta \log_q(\delta) - (1-\delta) \log_q(1-\delta).
\end{equation*}
This bound is achieved by a random linear code (assuming a prime power
alphabet size) with overwhelming probability, and one can also make
sure that the code is balanced, by forcing the all-ones word to be in
the code.  When $\delta=1-(1+\eps)/q$, the bound $1-h_q(\delta)$
becomes
\begin{equation} \label{eq:GVcritical} \frac{\eps^2}{2(q-1)\ln q}+
  \frac{\eps^3(q-2)}{6(q-1)^2\ln q}+O_q(\eps^4) = \Omega(\eps^2/(q
  \log q)).
\end{equation}

Now let us instantiate the above results with a balanced $q$-ary code
$\C \subseteq \Z_q^n$ on the Gilbert-Varshamov bound and with relative
minimum distance $1-(1+\eps)/q$. First, consider the spherical
encoding $\Sph(\cOne)$ and suppose that we wish to obtain an $n\times
N$ $\RIP{2}$ matrix of order $L$ with a fixed constant $\alpha$. In
order to apply Corollary~\ref{coro:sphRIP}, we need to set $\eps =
\alpha/(2L)$. In this case, the Gilbert-Varshamov bound implies that
the rate $R$ of $\C$ can be at least $\Omega(\eps^2/(q \log q))=
\Omega(\alpha^2/(L^2 q \log q))$.  The number of columns of the
resulting matrix is $N = q^{Rn-1}$.  Therefore, we have
\[
\log N = (Rn-1) \log q = \Omega(\alpha2 n/(L^2 q)),
\]
or in other words,
\begin{equation} \label{eqn:sphRows} n = O(L^2 (\log N)
  q/\alpha^2)=O_{\alpha,q}(L^2 \log N).
\end{equation}
We remark that Porat and Rothschild \cite{ref:PR08} show how to
derandomize the probabilistic construction of linear codes on the
Gilbert-Varshamov bound for any fixed prime power alphabet $q$. They
design a deterministic algorithm for constructing the generator matrix
of the code in time $O(n q^{Rn})$, where $R$ is the rate\footnote{ The
  algorithm can be adapted to ensure that the obtained code is
  balanced.}. This running time is in nearly \emph{linear} in the
number of the entries of the resulting RIP matrix.

It is well known that there are $\RIP{2}$ matrices of order $L$ with
$n=O(L \log(N/L))$ rows and this bound is achieved by several
probabilistic constructions (in particular, independent Bernoulli $\pm
1/\sqrt{n}$ entries) \cite{ref:Kas77,ref:CRT06}.  However we see that
even using codes on the Gilbert-Varshamov bound the number of rows of
the RIP matrix obtained from Corollary~\ref{coro:sphRIP} becomes
larger by a multiplicative factor of about $\Omega(L)$. To see whether
this can be improved, we consider negative bounds on the rate-distance
trade-offs of codes.

For our range of parameters, the best known negative bounds on the
rate-distance of error-correcting codes (that show upper bounds on the
rate of any code with a certain minimum distance) are given by
linear-programming techniques. In particular, the linear programming
bound due to McEliece, Rodemich, Rumsey, and Welch
(cf.~\cite[Chapter~5]{vanlint}) states that, asymptotically, any
binary code with relative minimum distance at least $\delta$ and rate
$R$ must satisfy
\[
R \leq h(1/2-\sqrt{\delta(1-\delta)})+o(1).
\]
This bound can be generalized to $q$-ary codes as follows (see
\cite{ref:NS09}).
\begin{equation} \label{eqn:codeNeg} R \leq h_q\left( \frac{1}{q}
    (q-1-(q-2)\delta-2\sqrt{(q-1)\delta(1-\delta)}) \right) + o(1).
\end{equation}
For any fixed $q$, and for $\delta = 1-(1+\eps)/q$, this bound
simplifies to $R = O(\eps^2 \log(1/\eps))$. Using simple calculations
as before, we conclude that the $\RIP{2}$ matrix construction of
Corollary~\ref{coro:sphRIP} always requires $n=\Omega(L^2 (\log N) /
\log L)$ rows, regardless of the code being used.

The RIP matrices constructed from Corollary~\ref{coro:sphRIP} require
a factor $\tilde{\Omega}(L^2)$ in the number of rows due to the fact
that their column set forms a spherical code. It is known that any
$\eps$-coherent spherical code of size $N$ over $\CC^n$ must satisfy
the following (cf.~\cite{ref:Lev83})
\begin{equation} \label{eqn:coherenceNeg} \eps^2 = \Omega\left(
    \frac{\log N}{n \log(n/\log N)}\right),
\end{equation}
which implies $n=\Omega((\log N)/(\eps^2 \log(1/\eps)))$.  Therefore,
the factor $\eps^2$ in the denominator of the bound on $n$ (which
translates to a factor $L^2$ in the RIP setting) is necessary.

On the positive side, the reduction above from the codes on the
Gilbert-Varshamov bound indirectly shows that spherical codes with
coherence $\eps = O((\log N)/n)$ (i.e., $n=O(\eps^2 \log N)$) exist
and can be attained using probabilistic constructions. On the negative
side, the lower bound \eqref{eqn:coherenceNeg} can be translated
(using the reduction from error-correcting codes to spherical codes)
to upper bounds on the attainable rates of $q$-ary codes with distance
close to $1-1/q$. This results in an indirect upper bound comparable
to what the linear programming bound \eqref{eqn:codeNeg} implies.

Now we turn to the construction of RIP matrices from the Boolean
embedding of error-correcting codes obtained in
Corollary~\ref{coro:boolRIP}.  In order to obtain an $\RIP{2}$ matrix
of order $L$ with constant $\alpha$, by Corollary~\ref{coro:boolRIP}
it suffices to have a code $\C \subseteq \Z_q^n$ attaining the
Gilbert-Varshamov bound with relative minimum distance at least
$1-(1+\eps)/q$ and $\eps \leq (\alpha q/L)-1$.  For a fixed constant
$\alpha$, we can set $q=O(L)$ large enough (e.g., $q=2 L/\alpha$) and
choose $\eps$ to be a small absolute constant (e.g., $\eps=.01$) so
that the above condition is satisfied.  The resulting matrix would
have $N:=|\C|$ columns and $n' := nq$ rows, with entries that are
either $0$ or $1/\sqrt{n}$. Moreover, the matrix is rather
\emph{sparse} in that all but a $1/q$ fraction of the entries are
zeros.

Now, the Gilbert-Varshamov bound \eqref{eqn:GV} implies that the rate
$R$ of $\C$ can be made at least $\Omega(\eps^2/(q \log
q))=\Omega(1/(q \log q))$.  Thus we have
\[
\log N = \log |\C| = (R n'/q) \log q = \Omega(n'/q^2)
\]
which gives $n' = O(q^2 \log N) = O(L^2 \log N)$. This is comparable
to the bound \eqref{eqn:sphRows} that we obtained from spherical
embedding of codes.  Similar to the case of spherical codes, Boolean
embedding allows us to translate positive bounds on the rate-distance
trade-off of codes (e.g., the Gilbert-Varshamov bound) to upper bounds
on the coherence of spherical codes as well as upper bounds on the
number of rows of $\RIP{2}$ matrices. Conversely, through Boolean
embedding, lower bounds on the coherence of spherical codes and lower
bounds on the number of rows of $\RIP{2}$ matrices translate into
impossibility bounds on the rate-distance trade-off of
error-correcting codes, the former being comparable to the linear
programming bound \eqref{eqn:codeNeg} when the relative minimum
distance is around $1-1/q$, but the latter is much weaker (namely,
comparable to the Plotkin bound on codes \cite[Chapter~5]{vanlint}
which is, over small alphabets, much weaker than the linear
programming bounds).

\section{From Average Distance to RIP} \label{sec:avgToRIP}

As we saw in the previous section, the quadratic dependence on
sparsity $L$ is unavoidable when the column set of an $\Rip$ matrix
forms a low-coherence spherical code. In this section we introduce the
notion of \emph{$L$-wise distance} that turns out to be more closely
related to the $\Rip$.

\begin{defn} \label{def:AvgDist} Let $c_1, \ldots, c_L \in \Z_q^n$ be
  $L$ vectors. The average distance of $c_1, \ldots, c_L$ is defined
  in the natural way
  \[
  \dist_L(c_1, \ldots, c_L) = \frac{1}{n\binom{L}{2}} \left\{ \sum_{1
      \leq i < j \leq L} \Delta(c_i, c_j) \right\},
  \]
  where $\Delta(c_i, c_j)$ is the Hamming distance between $c_i$ and
  $c_j$.
\end{defn}

\begin{defn} \label{def:LwiseDist} Let $\C \subseteq \Z_q^n$ be a
  code, and $L$ be an integer where $1 < L \leq |\C|$. Define the
  \emph{$L$-wise distance} of $\C$ as
  \[
  \dist_L(\C) := \min_{ \{c_1, \ldots, c_L\} \subseteq \C}
  \dist_L(c_1, \ldots, c_L).
  \]
\end{defn}
The special case $L=2$ is equal to the minimum relative distance of
the code. For the other extreme case, $L = |\C|$, the $L$-wise
distance of the code is the average relative distance over all
codeword pairs. For linear codes, this quantity is the expected
relative weight of a random codeword, given by
\[
\dist_{|\C|}(\C) = \frac{(q-1)|\{i \in [n]\colon (\exists (c_1,
  \ldots, c_n) \in \C), c_i \neq 0 \}|}{q n}.
\]
Thus, as long as the code is non-constant at all positions, its
$|\C|$-wise distance is equal to $(1-1/q)$. Also, a simple exercise
shows the ``monotonicity property'' that for any code $\C$, and $L'
\geq L$, $\dist_{L'}(\C) \geq \dist_{L}(\C)$.

We will use the notion of \emph{flat RIP} below from
\cite{ref:BDFKK11}.

\begin{defn} \label{def:flatRIP} Let $\alpha > 0$ be a real parameter.
  An $n \times N$ matrix $M \in \CC^{n \times N}$ with columns $M_1,
  \ldots M_N \in \CC^n$ is said to satisfy \emph{flat} $\Rip$ of order
  $L$ with constant $\alpha$ if for all $i \in [N]$, $\|M_i\|=1$ and
  moreover, for any disjoint $\cL_1, \cL_2 \subseteq [N]$ with
  $|\cL_1| = |\cL_2| \leq L$ we have
  \[
  \left| \Iprod{\sum_{i \in \cL_1} M_i}{\sum_{i \in \cL_2} M_i}
  \right| \leq \alpha \sqrt{|\cL_1| |\cL_2|} =\alpha |\cL_1|.
  \]
\end{defn}

The original definition of flat RIP in \cite{ref:BDFKK11} is stronger
and does not assume the two sets $|\cL_1|$ and $|\cL_2|$ have equal
sizes. However, adding the extra constraint does not affect the result
that we use from their work (Lemma~\ref{lem:flatRIP} below).

A straightforward exercise shows that the standard $\RIP{2}$ is no
weaker than the flat $\Rip$, namely,

\begin{prop} \label{prop:RIPtoFlat} Suppose a matrix $M$ satisfies
  $\RIP{2}$ of order $2L$ with constant $\alpha$.  Then, $M$ satisfies
  flat $\Rip$ of order $L$ with constant $O(\alpha)$.
\end{prop}

More interestingly, the two notions turn out to be essentially
\emph{equivalent} (up to a logarithmic loss in the $\Rip$ constant) in
light of the following result by Bourgain et al.:

\begin{lem} \label{lem:flatRIP} \cite{ref:BDFKK11} Let $L \geq 2^{10}$
  and suppose that a matrix $M$ satisfies flat $\Rip$ of order $L$
  with constant $\alpha$. Then $M$ satisfies $\RIP{2}$ of order $2L$
  with constant $44 \alpha \log L$.
\end{lem}

The notion of $L$-wise distance is a relaxed variation of the minimum
distance, where the distance is averaged over various choices of $L$
distinct codewords, as opposed to only two. Similarly, the notions of
$\eps$-biased codes and spherical codes can be relaxed to $L$-wise
forms and one can obtain various generalizations of the results
presented in Section~\ref{sec:mindist} to codes satisfying the relaxed
notion of $L$-wise distance.

For clarity of presentation, for the remainder of this section we only
focus on binary codes.  In this case, if the code $\C$ with $L$-wise
distance at least $1/2-\eps$ contains the all-ones word, one can
simply show that not only the average distance of any choice of $L$
codewords in $\cOne$ is at least $1/2-\eps$, but this quantity is also
no more than $1/2+\eps$ (to see this, it suffices to note that the
average distance of $L$ codewords plus the average distance of their
negations equals one). Let us call codes satisfying this stronger
property \emph{$L$-wise $\eps$-biased}:

\begin{defn} \label{def:epsBias} Let $\C \subseteq \Z_2^n$ be a code,
  and $L$ be an integer where $1 < L \leq |\C|$. Then, $\C$ is called
  $L$-wise $\eps$-biased if
  \[
  \max_{ \{c_1, \ldots, c_L\} \subseteq \C} |\dist_L(c_1, \ldots,
  c_L)-1/2| \leq \eps.
  \]
\end{defn}

The result below shows how the flat RIP and $L$-wise distance are
related. Again, the result is only presented for binary codes and the
extension to $q$-ary codes is straightforward.

\begin{lem} \label{lem:flatRIPfromCode} Suppose $\C \subseteq \Z_2^n$
  is such that, for a positive integer $L_0$ and all $L \leq 2 L_0$,
  $\C$ is $L$-wise $(\alpha/L)$-biased.  Then, $\Sph(\C)$ satisfies
  flat $\Rip$ of order $L_0$ with constant $4 \alpha$.
\end{lem}

\begin{proof}
  Fix any $L' \leq 2 L_0$ and any collection $c_1, \ldots, c_{L'}$ of
  the codewords in $\C$. Define
  \begin{eqnarray}
    \eta(c_1, \ldots, c_{L'}) &:=& \sum_{1 \leq i <j \leq L'} \innr{\Sph(c_i), \Sph(c_j)} \nonumber\\
    &=& \sum_{1 \leq i <j \leq L'} (2\Delta(c_i, c_j)/n-1) \nonumber \\
    &\in& [-2\frac{\alpha}{L'} \binom{L'}{2}, +2 \frac{\alpha}{L'} \binom{L'}{2}] \label{eqn:flatRIPbiasFirst} \\
    &\in& [-\alpha L', +\alpha L']\label{eqn:flatRIPbias},
  \end{eqnarray}
  where \eqref{eqn:flatRIPbiasFirst} is due to the small-bias
  assumption on $\C$.

  Now, let $L \leq L_0$ and $M_1, \ldots, M_{2L}$ be distinct columns
  of $\Sph(\cOne)$ corresponding to distinct codewords $c'_1, \ldots,
  c'_{2L}$ in $\C$.  Now, from Definition~\ref{def:flatRIP}, we need
  to bound the quantity
  \begin{eqnarray*}
    \eta' &:=& \Innr{\sum_{i \in [L]} M_i, \sum_{L<i\leq 2L} M_i} \\
    &=& \sum_{1\leq i<j \leq 2L} \innr{M_i, M_j} - \sum_{1\leq i<j \leq L} \innr{M_i, M_j} - \\ && \sum_{L+1\leq i<j \leq 2L} \innr{M_i, M_j}.
  \end{eqnarray*}
  Now, the absolute value of $\eta'$ can be bounded as
  \begin{eqnarray*}
    |\eta'| &\leq& |\eta(c_1, \ldots, c_{2L})|+|\eta(c_1, \ldots, c_L)|+\\
    && |\eta(c_{L+1}, \ldots, c_{2L})| \stackrel{(\star)}{\leq} 4\alpha L,
  \end{eqnarray*}
  where $(\star)$ is from \eqref{eqn:flatRIPbias}.
\end{proof}

Note that, contrary to Corollary~\ref{coro:sphRIP}, the above result
does not require the code to have an extremal minimum distance. In
principle, $\C$ can have a minimum distance bounded away from $1/2$ by
a constant (depending on the constant $\alpha$) and still satisfy the
conditions of Lemma~\ref{lem:flatRIPfromCode}.

The above result is also valid in the reverse direction, as follows.

\begin{lem} \label{lem:FlatRIPtoCode} Let $M$ be an $n\times N$ matrix
  with entries in $\{-1/\sqrt{n},+1/\sqrt{n}\}$ satisfying the flat
  $\Rip$ of order $L_0$ with constant $\alpha$. Then, columns of $M$
  form the spherical encoding of a code $\C \subseteq \Z_2^n$ such
  that for any $L \leq L_0$, the code $\C$ is $L$-wise
  $O(\alpha/L)$-biased.
\end{lem}

\begin{proof}
  Assume $L$ is even (the odd case is similar). Consider any $L$
  distinct columns $M_1, \ldots, M_L$ of $M$ and observe that
  \begin{eqnarray*}
    \eta &:=& \sum_{1 \leq i<j \leq [L]}\innr{M_i, M_j}  \\ &=& \sum_{\substack{\cL \subseteq [L] \\ |\cL| = \frac{L}{2}}}
    \sum_{\substack{i \in \cL \\ j \in [L]\setminus \cL}} \innr{M_i, M_j}/\binom{L}{\frac{L}{2}-1}.
  \end{eqnarray*}
  By the flat $\Rip$, each term $\sum_{i \in \cL}\sum_{j \in
    [L]\setminus \cL} \innr{M_i, M_j}$ is upper bounded in absolute
  value by $\alpha L/2$, and therefore, the above equation simplifies
  in absolute value to $ |\eta| = O(\alpha L).  $ Now suppose the
  codewords corresponding to $M_1, \ldots, M_L$ are $c_1, \ldots,
  c_L$.  The $L$-wise distance of these codewords can be written as
  \begin{eqnarray*}
    \dist_L(c_1, \ldots, c_L)&=&\frac{1}{\binom{L}{2}}\left(\sum_{1\leq i<j\leq L}\frac{1+\innr{M_i,M_j}}{2}\right)\\
    &=& \frac{1}{2} + \eta/\binom{L}{2}.
  \end{eqnarray*}
  Hence, $ |\dist_L(c_1, \ldots, c_L)-1/2| = |\eta|/\binom{L}{2} =
  O(\alpha/L).  $
\end{proof}

\section{Designs and Disjunct Matrices} \label{sec:design}

In this section we turn to the problem of combinatorial group testing,
and in particular discuss coding-theoretic constructions of disjunct
matrices. One of the foremost constructions dates back to the work of
Kautz and Singleton \cite{ref:KS64}, who used Reed-Solomon codes for
the purpose of constructing disjunct matrices\footnote{The work of
  Kautz and Singleton aims to construct superimposed codes, which are
  closely related to disjunct matrices.}.  This work results in a
general framework for construction of disjunct matrices through
\emph{combinatorial designs}, which are defined as follows.

\begin{defn}
  An \emph{$(n,n',r)$-design} is a set system $S_1, \ldots, S_N
  \subseteq [n]$ such that the size of each set is $n'$ and for every
  pair $i, j \in [N]$ ($i \neq j$) we have $|S_i \cap S_j| \leq r$.
\end{defn}

The following simple observations show that designs can be used to
construct disjunct matrices, and can in turn be obtained from
error-correcting codes:

\begin{lem} \label{lem:designToDisjunct} Let $\cD=\{S_1, \ldots,
  S_N\}$ be an $(n,n',r)$-design, and consider the binary $n\times N$
  matrix $M$ induced by $\cD$ where the $i$th column of $M$ is
  supported on $S_i$. Then, $M$ is $L$-disjunct provided that $Lr <
  n'$.
\end{lem}

\begin{proof}
  It suffices to observe that in Definition~\ref{def:disjunct}, each
  of the $M_i$ for $i \in [L]$ contains at most $r$ of the $n'$
  elements on $\supp(M_0)$.
\end{proof}

\begin{lem} \label{lem:codeToDesign} Let $\C =\{ c_1, \ldots, c_N \}
  \subseteq \Z_q^{n'}$ be a code with minimum Hamming distance at
  least $d$.  For $n:=n' q$, consider the set system $\cD := \{
  S_i\colon i \in [N] \}$ defined from the Boolean embedding of $\C$
  as follows: $S_i := \supp(\Bool(c_i))$.  Then, $\cD$ is an $(n, n',
  n'-d)$-design.
\end{lem}

\begin{proof}
  Observe that intersection size $|S_i \cap S_j|$, for $i \neq j$, is
  equal to $n'-\Delta(c_i, c_j) \leq n'-d$. The rest of the conditions
  are trivial.
\end{proof}

Now let us instantiate the above lemmas with a $k$-dimensional
Reed-Solomon code, as in \cite{ref:KS64}.  In this case, the alphabet
size $q$ can be made equal to the block length $n'$ (assuming that
$n'$ is a prime power). From Lemma~\ref{lem:codeToDesign}, the
resulting $(n,n',r)$-design satisfies $n=n'^2$, $r=n'-(n'-k)=k$ (since
the minimum distance of the code is $n'-k+1$), and $\log N=k \log q=r
\log n'>r$.  Furthermore, by Lemma~\ref{lem:designToDisjunct},
characteristic vectors of the resulting design form a disjunct matrix
with sparsity parameter $L \approx n'/r$.  Therefore, the number of
rows $n$ can be upper bounded as $n=n'^2 \approx (rL)^2 < (L \log
N)^2$.

As a second example, consider choosing a $q$-ary code on the
Gilbert-Varshamov bound with minimum Hamming distance at least $d :=
n'-(1+\eps)n'/q$, for some small (and fixed) constant $\eps > 0$.
Recall that the rate $R$ of the code satisfies $R=\Omega(\eps^2/(q
\log q))$.  This time, we obtain an $(n,n',r)$-design with
$r=n'-d=(1+\eps)n'/q$, $n=n'q=(1+\eps)n'^2/r=O(n'^2/r)$ and $\log N =
Rn' \log q=\Omega(\eps^2/q)=\Omega(r \eps^2/(1+\eps))=\Omega(r)$.  Now
lemma~\ref{lem:designToDisjunct} implies that the measurement matrix
that has the Boolean embedding of the codewords as its columns is
$L$-disjunct for $L \approx n'/r$.  Note that since $q=(1+\eps)n'/r$,
we must choose $q=\Omega(L)$ for the bounds to follow.  Altogether, we
obtain $n=n'q=O(n'^2/r)=O(L^2 r)=O(L^2 \log N)$.

Probabilistic arguments can be used to show that $(n,n',r)$-designs of
size $N$ exist for $n=O(n'^2 N^{1/r}/r)$, and moreover, this bound is
known to be nearly tight (cf.\ \cite{ref:EFF85} and
\cite[Ch.~7]{ref:groupTesting}). Therefore, we see that the design
obtained from codes on the Gilbert-Varshamov bounds for which
$nr/{n'^2}=O(1)$ and $\log N = \Omega(r)$ essentially achieves the
best possible bounds.

Regarding the existence of disjunct matrices, it is known that
$L$-disjunct matrices exists with $n=O(L^2 \log N)$ rows (using the
probabilistic method) and moreover, any $L$-disjunct matrix must
satisfy $n=\Omega(L^2 \log_L N)$ (cf.\
\cite[Ch.~7]{ref:groupTesting}).  Again, we see that the disjunct
matrices obtained from codes on the Gilbert-Varshamov bounds are
essentially optimal. Moreover, such matrices can be generated in
polynomial time in the size of the matrix using the result of Porat
and Rothschild \cite{ref:PR08}.

\section{List Decoding and Sparse Recovery} \label{sec:listdecoding}


As we saw in Section~\ref{sec:avgToRIP}, the relaxed notion of
$L$-wise distance essentially captures the $\RIP{2}$ for matrices with
$\pm 1/\sqrt{n}$ entries.  In this section, we relate this notion to
the standard notion of combinatorial list-decoding that has been
extensively studied in the coding-theory literature.

We remark that the notion of soft-decision list-decodable codes has
been used for construction of $\RIP{1}$ matrices, and it is known that
optimal $\RIP{1}$ matrices can be constructed from optimal
\emph{soft-decision} list-decodable codes which, in particular, imply
optimal unbalanced lossless expander graphs (see
\cite{ref:IR08,ref:BGIKS08,ref:Sina} and the references therein for
the construction of $\RIP{1}$ matrices from expander graphs and
\cite{ref:Vad10} for the reduction from codes to expander graphs). The
goal is this section is to show how list-decoding is related to the
more geometric property $\RIP{2}$.

\begin{defn} \label{def:listdecode} A code $\C \subseteq \Z_q^n$ is
  $(L, \rho)$-list decodable if for any $x \in \Z_q^n$, we have $
  |\cB(x, \rho) \cap \C| < L, $ where $\cB(x, \rho)$ denotes the
  Hamming ball of radius $\rho n$ around $x$.
\end{defn}

In the following lemma, we show that codes with good $L$-wise distance
have good list-decoding properties.

\begin{lem} \label{lem:JohnsonExt} Suppose that the $L$-wise distance
  of a code $\C \subseteq \Z_2^n$ is at least $1/2 - \eps^2$, where
  $L=O(1/\eps^2)$. Then, $\C$ is $(O(1/\eps^2), 1/2-\eps)$-list
  decodable.
\end{lem}

\begin{proof}
  The proof idea is inspired by a geometric proof of the Johnson's
  bound due to Guruswami and Sudan \cite{ref:GS01}. By the end of the
  proof, we will determine an $L' = O(1/\eps^2)$ satisfying $L' \geq
  L$ such that the assumption that $\C$ is not $(L',\eps)$-list
  decodable leads to a contradiction.

  Now, for the sake of contradiction, consider any $x \in \Z_2^n$ for
  which $\C \cap \cB(x, 1/2-\eps)$ has size at least $L'$. Take any
  set of distinct codewords
  \[ c_1, \ldots, c_{L'} \in \C \cap \cB(x, 1/2-\eps)\] and consider
  the spherical encodings $v_0 := \Sph(x)$, $v_1 := \Sph(c_1), \ldots,
  v_\ell := \Sph(c_{L'}).$ By the monotonicity property of the
  $L$-wise distance, we know that $\dist_{L'}(c_1, \ldots, c_{L'})
  \geq 1/2 - \eps^2.$ For spherical embeddings, this translates to
  \begin{equation}
    \sum_{1 \leq i<j \leq L'} \innr{v_i, v_j}=\binom{L'}{2}(1-2\dist_{L'}(c_1, \ldots, c_{L'})) \leq 2 L'^2 \eps^2 \label{eqn:johnsonAvg}
  \end{equation}
  Also, since the relative Hamming distance between $x$ and any $c_i$
  is at most $1/2-\eps$, we get
  \begin{equation}
    (\forall i \in [L']) \quad \innr{v_i, v_0}=(1-2 \Delta(c_i, x)/n) \geq 2 \eps. \label{eqn:johnsonDist}
  \end{equation}
  Using \eqref{eqn:johnsonDist}, for every $i \in [L']$ and parameter
  $\beta > 0$,
  \begin{equation}
    \innr{v_i-\beta v_0,v_i-\beta v_0} = 1+\beta^2-2\beta \innr{v_i,v_0} \leq 1+\beta^2-4\eps \beta. 
    \label{eqn:johnsonEqual}
  \end{equation}
  Similarly, for $1 \leq i<j \leq L'$ we can write
  \begin{eqnarray}
    \innr{v_i-\beta v_0,v_j-\beta v_0} &=& \innr{v_i,v_j}+\beta^2-\beta \innr{v_i+v_j,v_0} \nonumber \\
    &\leq& \innr{v_i,v_j}+\beta^2-4\eps \beta.
    \label{eqn:johnsonNotEqual}
  \end{eqnarray}
  Altogether,
  \begin{eqnarray}
    0 &\leq& \Innr{\sum_{i\in[L']} (v_i-\beta v_0),\sum_{i\in[L']} (v_i-\beta v_0)} \nonumber \\
    &=& \sum_{i\in[L']} \innr{v_i-\beta v_0,v_i-\beta v_0} +\nonumber \\
    && \sum_{1\leq i<j \leq L'} \innr{v_i-\beta v_0,v_j-\beta v_0}\nonumber \\
    &\leq&
    L'(1+\beta^2-4\eps \beta) + \nonumber \\
    && 2 L'^2 \eps^2 + (L'^2-L')(\beta^2-4\eps \beta), \nonumber 
  \end{eqnarray}
  where the last inequality is using \eqref{eqn:johnsonAvg},
  \eqref{eqn:johnsonEqual}, and
  \eqref{eqn:johnsonNotEqual}. Therefore, after reordering, we have $
  L' \leq 1/(4\eps \beta -\beta^2 -2\eps^2), $ provided that the
  denominator is positive. Now we choose $\beta := \eps$ to get $L'
  \leq 1/\eps^2$. Therefore, it suffices to choose $L' > \max\{
  1/\eps^2, L \}$ to get the desired contradiction.
\end{proof}

A sequence of results that we have seen so far can be combined to
obtain list-decodable codes from $\Rip$ matrices. Namely, starting
from a binary $\Rip$ matrix, we can apply
Proposition~\ref{prop:RIPtoFlat}, Lemma~\ref{lem:FlatRIPtoCode}, and
Lemma~\ref{lem:JohnsonExt} in order and obtain the following:

\begin{lem} \label{lem:RIPtoLD} Suppose an $n \times N$ matrix $M$
  with entries in $\{-1/\sqrt{n}, +1/\sqrt{n}\}$ satisfies the
  $\RIP{2}$ of order $L$ with constant $\alpha$. Let $\C \subseteq
  \Z_2^n$ be the binary code such that $M=\Sph(\C)$. Then, there is a
  parameter $\eps_0 = O(\sqrt{\alpha/L})$ such that for every $\eps
  \geq \eps_0$, $\C$ is $(O(1/\eps^2), 1/2-\eps)$-list decodable.
\end{lem}

Recall that the probabilistic method shows that $\RIP{2}$ matrices of
order $L$ exist with $N$ columns and $n=O(L \log(N/L))$ rows, and this
is achieved with overwhelming probability by a random matrix (with
$\pm 1/\sqrt{n}$ entries). Using such a matrix in the above lemma, we
obtain an $(O(1/\eps^2), 1/2-\eps)$-list decodable code with rate
$R=\Omega(\eps^2)$.  It can be directly shown that this list-decoding
trade-off is achieved by random codes with overwhelming probability,
and the trade-off is essentially optimal (cf.\ \cite{ref:GHSZ02}).
However, explicit construction of optimal $\RIP{2}$ matrices and
optimal binary list-decodable codes at radius $1/2-\eps$ are both
challenging open problems.  Therefore, Lemma~\ref{lem:RIPtoLD} relates
two important explicit construction problems; namely, it implies a
reduction from Problem~\ref{prob:LD} to Problem~\ref{prob:RIP}
below\footnote{ We remark that, for the reduction to yield explicit
  list-decodable codes, an explicit algorithm that computes the $\Rip$
  matrix in polynomial time in the size of the matrix would not
  necessarily suffice. One needs the more stringent explicitness that
  requires each individual entry of the matrix to be computable in
  time $\poly(n)$.  } (when the latter problem is restricted to binary
real matrices).

\begin{prob} \label{prob:LD} Construct an explicit family of binary
  codes with block length $n$ and rate $R=\Omega(\eps^2)$ that are
  $(O(1/\eps^2), 1/2-\eps)$-list decodable.
\end{prob}

\begin{prob} \label{prob:RIP} Construct an explicit family of
  $\RIP{2}$ matrices of order $L$ with $N$ columns and $n=O(L
  \log(N/L))$ rows.
\end{prob}

In Section~\ref{sec:mindist} we showed how to obtain explicit
$\RIP{2}$ matrices from spherical embedding of codes on the
Gilbert-Varshamov bound constructed by Porat and Rothschild
\cite{ref:PR08}. This construction achieves $n=O(L^2 \log N)$, which
achieves the best known explicit bound for matrices with $\pm
1/\sqrt{n}$ entries\footnote{Bourgain et al.~\cite{ref:BDFKK11}
  explicitly obtain a better-than-quadratic dependence on $L$ for an
  interesting range of parameters. However, entries of their matrices
  are powers of the primitive complex $p$th root of unity for a large
  prime $p$.}. Observe that the dependence on $L$ is sub-optimal by a
factor two in the exponent.  As for binary list-decodable codes at
radius close to $1/2$ (and small list-size), Guruswami et al.\
construct explicit $(O(1/\eps^2), 1/2-\eps)$-list decodable codes of
rate $R=\Omega(\eps^4)$ \cite{ref:GHSZ02}. Again, the exponent of
$\eps$ in the rate is sub-optimal by a factor two.

A natural question is whether the reduction offered by
Lemma~\ref{lem:RIPtoLD} holds in the reverse direction as well;
namely,

\begin{ques} \label{que:RIPfromLD} Let $\C \subseteq \Z_2^n$ be such
  that, for some integer $L$ and every $1 < L' \leq L$, the code $\C$
  is $(L', 1/2-O(\sqrt{\alpha/L'}))$-list decodable.  Does $\Sph(\C)$
  satisfy $\RIP{2}$ of order $\Omega(L)$ with constant $O(\alpha)$?
\end{ques}

From Lemmas \ref{lem:flatRIPfromCode}~and~\ref{lem:flatRIP} we know
that in order to answer the above question in affirmative, it suffices
to show a converse to Lemma~\ref{lem:JohnsonExt}.  A weak converse,
not strong enough for this purpose, is shown below.

\begin{lem} \label{lem:JohnsonConverse} Suppose that a code $\C
  \subseteq \Z_2^n$ is $(L, 1/2-\eps)$-list decodable.  Then, for $L'
  := L/\eps$, the $L'$-wise distance of the code $\C \subseteq \Z_2^n$
  is at least $1/2 - 2 \eps$.
\end{lem}

\begin{proof}
  The proof is, in essence, a straightforward averaging argument.
  Suppose, for the sake of contradiction, that there is a set of $L'$
  codewords whose average distance is less than $1/2 - 2 \eps$. Denote
  the spherical encodings of these codewords by $c_1, \ldots, c_{L'}$,
  each in $\{-1,+1\}^n$. From the definition of $L'$-wise distance
  (Definition~\ref{def:LwiseDist}), we have
  \begin{equation} \label{eqn:avgDistLarge} \sum_{1 \leq i < j < L'}
    \innr{c_i, c_j} > 4 \eps \binom{L'}{2} n > 2 \eps L'^2 n = 2 L^2 n
    / \eps.
  \end{equation}
  Now, define $v := \sum_{i=1}^{L'} c_i / L'$ and note that this is a
  real vector in $[-1,+1]^n$. Suppose $v = (v_1, \ldots, v_n)$ and
  randomly pick a vector $\bar{v} = (\bar{v}_1, \ldots, \bar{v}_n) \in
  \{-1, +1\}^n$ with independent coordinates such that $\E[\bar{v}_i]
  = v_i$. This is possible since each $v_i$ is in $[-1, +1]$.

  Note that, by linearity of expectation, for every $i$ we have
  $\E[\innr{\bar{v}, c_i}] = \innr{v, c_i}$.  Again, using linearity
  of expectation,
  \begin{eqnarray}
    \E[\sum_{i=1}^{L'} \innr{\bar{v}, c_i}] &=& 
    \sum_{i=1}^{L'} \innr{v, c_i} = \innr{v, \sum_{i=1}^{L'} c_i} \nonumber \\
    &=& \frac{1}{L'} \innr{\sum_{i=1}^{L'} c_i, \sum_{i=1}^{L'} c_i} \nonumber \\
    &=& \frac{1}{L'} (L' + 2 \sum_{1 \leq i < j \leq L'} \innr{c_i, c_j}) \nonumber \\
    &\stackrel{\eqref{eqn:avgDistLarge}}{>}& 1+ 4 L n > 4 L n. \label{eqn:largeExp}
  \end{eqnarray}
  Since there is a choice of the randomness that preserves
  expectation, we can ensure that there is a deterministic choice of
  $\bar{v} \in \{-1,+1\}^n$ that satisfies \eqref{eqn:largeExp}. In
  the sequel, fix such a $\bar{v}$. We thus have
  \begin{equation} \label{eqn:fixedChoice} \sum_{i=1}^L \innr{\bar{v},
      c_i} > 4Ln.
  \end{equation}

  Now, \eqref{eqn:largeExp} implies that there must be a set $S$ of
  more than $L$ vectors in $\{c_1, \ldots, c_{L'}\}$ such that for
  every $c \in S$, the inequality $\innr{v, c} > 2\eps n$ holds, since
  if this were not the case, we would have
  \[
  \sum_{i=1}^L \innr{\bar{v}, c_i} \leq L' (2\eps n) + L n = 3Ln,
  \]
  contradicting \eqref{eqn:fixedChoice}. We conclude that the set of
  codewords corresponding to the spherical encodings in $S$ are all
  $(1/2-\eps)$-close in Hamming distance to the binary vector
  represented by $\bar{v}$.  This contradicts the assumption that $\C$
  is $(L, 1/2-\eps)$-list decodable and completes the proof.
\end{proof}

\section{Conclusion} \label{sec:conclusion}

The reductions between coding-theoretic objects such as codes with
large distance, incoherent spherical codes, combinatorial designs and
the like are not only interesting for constructions, but also they
relate the known bounds on the parameters achievable by one to
another. For example, due to the reduction from binary codes to
spherical codes, any improved lower bound on the coherence of
spherical codes results in an improved upper bound on the rates
achievable by small-biased codes. Thus, it is interesting to explore
further connections of this type. For example, whether there is a
reduction from disjunct matrices to designs, designs to codes, etc.
Moreover, an affirmative answer to Question~\ref{que:RIPfromLD} would
imply that the seemingly unrelated problems of finding explicit
$\RIP{2}$ matrices (with $\pm 1/\sqrt{n}$ entries) and explicit binary
list-decodable codes\footnote{Note that there is no requirement on the
  existence of an efficient list-decoder for the code. Only the
  encoding function needs to be efficient.} at radius close to $1/2$
are essentially equivalent. In particular, optimal $\RIP{2}$ matrices
would imply optimal binary list-decodable codes and vice versa.  One
can also ask similar questions about non-binary codes, which might be
easier to construct, or consider related variations of the $L$-wise
distance\footnote{One possibility is to look at the spherical
  embedding of the code and work with average inner products of pairs
  within all collections of $L$ codewords, rather than the average
  $L$-wise distance as in Definition~\ref{def:AvgDist}.}.


\section*{Acknowledgements}

I would like to thank Venkatesan Guruswami, Sina Jafarpour, and David
Zuckerman for discussions related to the material presented in this
paper.

\bibliographystyle{IEEEtran} \bibliography{bibliography}

\end{document}